\documentclass{llncs}

\newif\ifblind\blindfalse
\newif\iffinal\finaltrue
\newif\iffull\fulltrue

\usepackage{amsmath, amssymb, stmaryrd, alltt}
\usepackage{braket}
\usepackage{empheq}
\usepackage{algorithm,algpseudocode}

\usepackage{color,bcprules,cite,wrapfig,url}

\newcommand{\Span}[1]{\langle{#1}\rangle}
\DeclareMathOperator{\WHILE}{\mathbf{while}}

\newcommand{\NAT}{\mathbb N}
\newcommand{\GRP}{\mathbb G}
\newcommand{\REAL}{\mathbb R}
\newcommand{\WPA}[2]{\llbracket #1 \rrbracket^{\sharp}_{#2}}
\newcommand{\WPAH}[2]{\llbracket {#1} \rrbracket^{\sharp \mathtt{H}}_{#2}}
\newcommand{\WPC}[2]{\llbracket #1 \rrbracket^{\sharp\mathtt{c}}_{#2}}
\newcommand{\WPCH}[2]{\llbracket #1 \rrbracket^{\sharp\mathtt{cH}}_{#2}}
\newcommand{\REM}{\mathbf{Rem}}

\newcommand{\REMP}{{\mathbf{Rem}^{\mathtt{par}}}}
\newcommand{\REMPH}[1]{{\mathbf{Rem}^{\mathtt{parH}}_{#1}}}
\newcommand{\SKIP}{\mathbf{skip}}
\newcommand{\ASSIGN}[2]{#1 \mathtt{:=} #2}
\newcommand{\SEQ}[2]{#1; #2}
\DeclareMathOperator{\IF}{\mathbf{if}}
\DeclareMathOperator{\THEN}{\mathbf{then}}
\DeclareMathOperator{\ELSE}{\mathbf{else}}
\newcommand{\IFTHENELSE}[3]{\mathbf{if}\, \allowbreak #1\,
  \allowbreak \mathbf{then}\,
  \allowbreak #2 \, \allowbreak \mathbf{else}\, \allowbreak #3}
\DeclareMathOperator{\WHL}{\mathbf{while}}
\DeclareMathOperator{\DO}{\mathbf{do}}
\newcommand{\SWHILE}[2]{\WHL\, #1\, \DO\, #2}
\newcommand{\PDEG}{\mathbf{deg}}
\newcommand{\DDEG}{\mathbf{gdeg}}
\newcommand{\LM}{\mathrm{LM}}

\newcommand{\INVINF}{\textsc{InvInf}}
\newcommand{\INVINFH}{\textsc{InvInf}^{\mathtt{H}}}

\newcommand{\HLE}{\preceq}

\makeatletter
\def\old@comma{,}
\catcode`\,=13
\def,{%
  \ifmmode%
  \old@comma\discretionary{}{}{}%
  \else%
    \old@comma%
  \fi%
}
\makeatother
\def\enumSet#1{\{\,#1\,\}}

\title{Generalized Homogeneous Polynomials for\\ Efficient
  Template-Based\\ Nonlinear Invariant Synthesis}

\ifblind
\author{}
\institute{}
\else
\author{Kensuke Kojima\inst{1,2} \and Minoru Kinoshita\inst{1,4} \and Kohei Suenaga\inst{1,3}}
\institute{Kyoto University \and JST CREST \and JST PRESTO \and KLab Inc.}
\fi

\newcommand\POLY{K}

\newcommand\VARS{\mathbf{Var}}
\newcommand\sem[1]{\llbracket #1 \rrbracket}

\newcommand\INT{\mathbb{Z}}
\newcommand\p{\vdash}

\newcommand\ra{\rightarrow}

\newcommand\DT{\mathit{dt}}

\iffinal
\newcommand\comment[1]{}
\newcommand\todo[1]{}
\else
\newcommand\comment[1]{\color{blue}{\bf{COMMENT: {#1} }}\color{black}}
\newcommand\todo[1]{\color{red}{\textbf{#1}}\color{black}}
\fi

\newcommand\CARD[1]{|{#1}|}

\newcommand\STATES{\mathbf{St}}

\newcommand\POWERSET{\mathcal{P}}

\newcommand\GDEGS{\mathbf{GDeg}}

\newcommand\EQCONSTR[2]{\langle #1 \equiv #2 \rangle}

\DeclareMathOperator\vorder{\sqsubseteq^\sharp}
\DeclareMathOperator\setdiff{\backslash}

\DeclareMathOperator\MOD{\mathbf{mod}}

\newcommand\FREEFALL{{\mathit{fall}}}

\newcommand\floor[1]{\lfloor {#1} \rfloor}

\newenvironment{pfof}[1]{\paragraph{Proof of {#1}}}{\qed}

\newcommand\PT{\mathit{PT}}
\newcommand\GDEGV{\mathbf{GDegV}}

\begin{document}

\maketitle

\allowdisplaybreaks

\begin{abstract}
  
 The \emph{template-based} method is one of the most successful
 approaches to algebraic invariant synthesis.  In this method, an
 algorithm designates a \emph{template polynomial} \(p\) over program
 variables, generates constraints for \(p=0\) to be an invariant, and
 solves the generated constraints.  However, this approach often suffers
 from an increasing template size if the degree of a template polynomial
 is too high.
  
 We propose a technique to make template-based methods more efficient.
 Our technique is based on the following finding: If an algebraic
 invariant exists, then there is a specific algebraic invariant that we
 call a \emph{generalized homogeneous} algebraic invariant that is often
 smaller.  This finding justifies using only a smaller template that
 corresponds to a generalized homogeneous algebraic invariant.

 Concretely, we state our finding above formally based on the abstract
 semantics of an imperative program proposed by Cachera et al.  Then, we
 modify their template-based invariant synthesis so that it generates
 only generalized homogeneous algebraic invariants.  This modification
 is proved to be sound.  Furthermore, we also empirically demonstrate
 the merit of the restriction to generalized homogeneous algebraic
 invariants.  Our implementation outperforms that of Cachera et al. for
 programs that require a higher-degree template.
\end{abstract}

\section{Introduction}
\label{sec:introduction}

We consider the following \emph{postcondition problem}: Given a program \(c\), discover a
fact that holds at the end of \(c\) regardless of the initial state.
This paper focuses on a postcondition written as an \emph{algebraic
condition} \(p_1 = 0 \wedge \dots \wedge p_n = 0\), where
\(p_1,\dots,p_n\) are polynomials over program variables; this problem
is a basis for static verification of functional correctness.

One approach to this problem is \emph{invariant synthesis}, in which we
are to compute a family of predicates \(P_l\) indexed by program
locations \(l\) such that \(P_l\) holds whenever the execution of \(c\)
reaches \(l\).  The invariant associated with the end of \(c\) is a
solution to the postcondition problem.


Because of its importance in static program verification, algebraic
invariant synthesis has been intensively
studied~\cite{DBLP:conf/popl/SankaranarayananSM04,DBLP:journals/ipl/Muller-OlmS04,DBLP:journals/jsc/Rodriguez-CarbonellK07,DBLP:journals/scp/CacheraJJK14}.
Among these proposed techniques, one successful approach is the
constraint-based method in which invariant synthesis is reduced to a
constraint-solving problem.  During constraint generation, this method
designates \emph{templates}, which are polynomials over the program
variables with unknown parameters at the coefficient
positions~\cite{DBLP:conf/popl/SankaranarayananSM04}.  The algorithm
generates constraints that ensure that the templates are invariants and
obtains the invariants by solving the constraints\footnote{The
constraint-based method by Cachera et
al.~\cite{DBLP:journals/scp/CacheraJJK14}, which is the basis of the
current paper, uses a template also for other purposes.  See
Section~\ref{sec:template} for details.}.


\begin{wrapfigure}{r}[0pt]{0.7\linewidth}
 \small
  \fbox{
  \begin{minipage}{\linewidth}
  \begin{center}
  \begin{algorithmic}[1]
    \State{\(x := x_0; v := v_0; t := t_0;\)}
    \While{\(t - a \ne 0\)}
    \State{\((x, v, t) := (x + v \DT, v - g \DT - \rho v \DT, t + \DT);\)}
    \EndWhile\label{line:loopend}
    \label{line:freefalllast}\State{}
  \end{algorithmic}
  \end{center}
  \end{minipage}
  } \caption{Program \(c_\FREEFALL\), which models a falling mass point.
  The symbols in the program represent the following quantities: \(x\)
  is the position of the point, \(v\) is its speed, \(t\) is time,
  \(x_0\) is the initial position, \(v_0\) is the initial speed, \(t_0\)
  is the initial value of the clock \(t\), \(g\) is the acceleration
  rate, \(\rho\) is the friction coefficient, and \(\DT\) is the
  discretization interval.  The simultaneous substitution in the loop
  body numerically updates the values of \(x\), \(v\), and \(t\).  The
  values of \(x\), \(v\), and \(t\) are numerical solutions of the
  differential equations \(\frac{dx}{\DT} = v\) and \(\frac{dt}{\DT} =
  1\); notice that the force applied by the air to the mass point is \(-
  \rho v\), which leads to the differential equation for
  \(\frac{dv}{\DT} = -g - \rho v\).}  \label{fig:freefall}
\end{wrapfigure}

\begin{example}
  \label{ex:freefall} The program \(c_\FREEFALL\) in
  Figure~\ref{fig:freefall} models the behavior of a mass point with
  weight \(1\) and with a constant acceleration rate; the program takes
  friction between the mass point and air into account\footnote{Although
  the guard condition \(t - a \ne 0\) should be \(t - a < 0\) in a
  real-world numerical program, we use the current example for
  presentation purposes.}.  For this program, the postcondition
  \(-gt+gt_0-v+v_0-x\rho+x_0\rho = 0\) holds regardless of the initial
  state.
\end{example}


We describe how a template-based method computes the postcondition in
Example~\ref{ex:freefall}.  The method described here differs from the
one we explore in this paper; this explanation is intended to suggest
the flavor of a template method.

A template-based method generates a \emph{template polynomial}\\ over the
program variables that represent an invariant at
Line~\ref{line:loopend}.  Suppose the generated polynomial
\(p(x_0,v_0,t_0,x,v,t,a,\DT,g,\rho)\) is of degree 2 over the variables:
\(p(x_0,v_0,t_0,x,v,\allowbreak t,a,\DT,g,\rho) := a_{1} + a_{t_0} t_0 +
a_{x_0} x_0 + \dots + a_{g\rho} g \rho\), where \(a_w\) is the
coefficient parameter associated with the power product \(w\).
The procedure then generates constraints such that
\(p(x_0,v_0,t_0,x,v,t,a,\DT,g,\rho) = 0\) is indeed an invariant at
Line~\ref{line:loopend}.
The method proposed by Sankaranarayanan et
al.~\cite{DBLP:conf/popl/SankaranarayananSM04} based on the Gr\"obner
basis~\cite{Cox:2007:IVA:1204670} generates the constraints as an
equations over the parameters; in this case, a solution to the
constraints gives \(-gt+ga-v+v_0-x\rho+x_0\rho = 0\), which is indeed an
invariant at the end of \(c_\FREEFALL\).

One of the drawbacks of the template-based method is excessive growth of
the size of a template.  Blindly generating a template of degree \(d\)
for a degree parameter \(d\) makes the invariant synthesis less scalable
for higher-degree invariants.  For example, the program in
Example~\ref{ex:freefall} has an invariant \(-gt^2 +gt_0^2 - 2tv +
2t_0v_0 + 2x - 2x_0 = 0\) at Line~\ref{line:loopend}.  This invariant
requires a degree-\(3\) template, which has ${10 + 3 \choose 3} = 286$
monomials in this case.

We propose a hack to alleviate this drawback in the template-based
methods.  Our method is inspired by a rule of thumb in physics called
the \emph{principle of quantity dimension}: A physical law should not
add two quantities with different \emph{quantity
dimensions}~\cite{barenblatt1996scaling}.  If we accept this principle,
then, at least for a physically meaningful program such as
\(c_\FREEFALL\), an invariant (and therefore a template) should consist
of monomials with the same quantity dimensions.

Indeed, the polynomial \(-gt+gt_0-v+v_0-x\rho+x_0\rho\) in the invariant
calculated in Example~\ref{ex:freefall} consists only of quantities that
represent velocities.  (Notice that \(\rho\) is a quantity that
corresponds to the inverse of a time quantity.)  The polynomial \(-gt^2
+gt_0^2 - 2tv +2t_0v_0 + 2x - 2x_0\) above consists only of quantities
corresponding to the length.  If we use the notation of quantity
dimensions used in physics, the former polynomial consists only of
monomials with the quantity dimension \(LT^{-1}\), whereas the latter
consists only of \(L\), where \(L\) and \(T\) represent quantity
dimensions for lengths and times, respectively.

By leveraging the quantity dimension principle in the template synthesis
phase, we can reduce the size of a template.  For example, we could use
a template that consists only of monomials for, say, velocity quantities
instead of the general degree-\(2\) polynomial
\(p(x_0,v_0,x,v,t,a,\DT,g,\rho)\) used above, which yields a smaller
template.

The idea of the quantity dimension principle can be nicely captured by
generalizing the notion of \emph{homogeneous polynomials}.  A polynomial is
said to be \emph{homogeneous} if it consists of monomials of the same degree;
for example, the polynomial \(x^3 + x^2 y + x y^2 + y^3\) is a
homogeneous polynomial of degree \(3\).  We generalize this notion of
homogeneity so that (1) a \emph{degree} is an expression corresponding
to a quantity dimension (e.g., \(LT^{-1}\)) and (2) each variable has
its own degree in degree computation.

Let us describe our idea using an example, deferring formal definitions.
Suppose we have the following \emph{degree assignment} for each program
variable: $\Gamma := \enumSet{x_0 \mapsto L, t_0 \mapsto T, g \mapsto
LT^{-2}, t \mapsto T, \DT \mapsto T, x \mapsto L, v \mapsto LT^{-1}, v_0
\mapsto LT^{-1}, \rho \mapsto T^{-1}, a \mapsto T}$.  This degree
assignment intuitively corresponds to the assignment of the quantity
dimension to each variable.  With this degree assignment \(\Gamma\), all
of the monomials in \(-gt+gt_0-v+v_0-x\rho+x_0\rho\) have the same
degree; for example, the monomial \(-gt\) has degree
\(\Gamma(g)\Gamma(t) = (LT^{-2}) T = LT^{-1}\) and monomial \(x\rho\)
has degree \(\Gamma(x)\Gamma(\rho) = L T^{-1}\), and so on.  Hence,
\(-gt+gt_0-v+v_0-x\rho+x_0\rho\) is a homogeneous polynomial in the
generalized sense.  Such a polynomial is called a \emph{generalized
homogeneous (GH) polynomial}.  We call an algebraic invariant with a GH
polynomial a \emph{generalized homogeneous algebraic (GHA) invariant}.

The main result of this paper is a formalization of this idea: If there
is an algebraic invariant of a given program \(c\), then there is a GHA
invariant.  This justifies the use of a template that corresponds to a
GH polynomial in the template method.  We demonstrate this result by
using the abstract semantics of an imperative programming language
proposed by Cachera et al.~\cite{DBLP:journals/scp/CacheraJJK14}.  We
also empirically show that the algorithm by Cachera et al. can be made
more efficient using this idea.

As we saw above, the definition of GH polynomials is parameterized over
a degree assignment \(\Gamma\).  The type inference algorithm for the
\emph{dimension type system} proposed by
Kennedy~\cite{DBLP:conf/esop/Kennedy94,kennedy96:_progr_languag_dimen}
can be used to find an appropriate degree assignment; \(\Gamma\) above
is inferred using this algorithm.  The dimension type system was
originally proposed for detecting a violation of the quantity-dimension
principle in a numerical program.  Our work gives an application of the
dimension type system to invariant synthesis.


Although the method is inspired by the principle of quantity dimensions,
it can be applied to a program that does not model a physical phenomenon
because we abstract the notion of a quantity dimension using that of
generalized homogeneity.  All the programs used in our experiments
(Section~\ref{sec:experiment}) are indeed physically nonsensical
programs.

The rest of this paper is organized as follows.
Section~\ref{sec:preliminary} sets up the basic mathematical definitions
used in this paper; Section~\ref{sec:language} defines the syntax and
semantics of the target language and its abstract semantics;
Section~\ref{sec:ghpoly} defines GH polynomials;
Section~\ref{sec:ghsemantics} defines the revised abstract semantics as
the restriction of the original one to the set of GH polynomials and
shows that the revised semantics is sound and complete;
Section~\ref{sec:template} gives a template-based invariant-synthesis
algorithm and shows its soundness; Section~\ref{sec:experiment} reports
the experimental results; Section~\ref{sec:relatedwork} discusses
related work; and Section~\ref{sec:conclusion} presents the conclusions.
Several proofs are given in \iffull the appendices \else the full
version~\cite{DBLP:journals/corr/KojimaKS16} \fi.

\section{Preliminaries}
\label{sec:preliminary}


\(\REAL\) is the set of real numbers and \(\NAT\) is the set of natural
numbers.  We write \(\CARD{S}\) for the cardinality of \(S\) if \(S\) is
a finite set.  We designate an infinite set of \emph{variables}
\(\VARS\).  \(K\) is a field ranged over by metavariable \(k\); we use
the standard notation for the operations on \(K\).  For \(x_1,\dots,x_n
\in \VARS\), we write \(\POLY[x_1,\dots,x_n]\), ranged over by \(p\) and
\(q\), for the set of polynomials in $x_1,\dots,x_n$ over $\POLY$.



A subset \(I \subseteq \POLY[x_1,\dots,x_n]\) is called an \emph{ideal}
if (1) \(I\) is an additive subgroup and (2) \(pq \in I\) for any \(p
\in I\) and \(q \in \POLY[x_1,\dots,x_n]\).  A set \(S \subseteq
\POLY[x_1,\dots,x_n]\) is said to \emph{generate} the ideal \(I\),
written $I = \Span{S}$,
if \(I\) is the smallest ideal that contains \(S\).

We call an expression of the form \(x_1^{d_1} \dots x_N^{d_N}\), where
\(d_1,\dots,d_N \in \NAT\) and \(x_1,\dots,x_N \in \VARS\), a
\emph{power product} over \(x_1,\dots,x_n\); \(w\) is a metavariable for
power products.  We call \(\sum d_i\) the \emph{degree} of this power
product.
A \emph{monomial} is a term of the form \(k w\); the degree of this
monomial is that of \(w\).
We write \(\PDEG(p)\), the degree of the polynomial \(p\),
for the maximum degree of the monomials in \(p\).


A \emph{state}, ranged over by \(\sigma\), is a finite map from
\(\VARS\) to \(K\).  We write \(\STATES\) for the set of states.  We
use the metavariable \(S\) for a subset of \(\STATES\).  We write
\(\sigma(p)\) for the evaluated value of \(p\) under \(\sigma\).
Concretely, \(\sigma(p) := p(\sigma(x_1),\dots,\sigma(x_n))\).
The set \(\POWERSET(\STATES)\)
constitutes a complete lattice with respect to the set-inclusion
order.

\section{Language}
\label{sec:language}

This section defines the target language, its concrete semantics, and
its abstract semantics.  We essentially follow the development by
Cachera et al.~\cite{DBLP:journals/scp/CacheraJJK14}; we refer the
interested reader to this paper.

The syntax of the target language is as follows:
\[
\footnotesize
\begin{array}{lcl}
  c &::=& \SKIP \mid \ASSIGN{x}{p} \mid \SEQ{c_1}{c_2} \mid \IFTHENELSE{p = 0}{c_1}{c_2} \mid \SWHILE{p = 0}{c} \mid \SWHILE{p \ne 0}{c}\\
\end{array}
\]
where \(p\) is a polynomial over the program variables.  We restrict the
guard to a single-polynomial algebraic condition (i.e., \(p = 0\)) or
its negation.

The semantics of this language is given by the following denotation
function, which is essentially the same as that by Cachera et al.
\[
\footnotesize
\begin{array}{rl}
  \sem{c} &: (\POWERSET(\STATES),\subseteq) \ra (\POWERSET(\STATES),\subseteq)\\
  \sem{\SKIP}(S) &= S\\
  \sem{\ASSIGN{x}{p}}(S) &= \set{\sigma \mid \sigma[x \mapsto \sigma(p)] \in S}\\
  \sem{\SEQ{c_1}{c_2}}(S) &= \sem{c_1}(\sem{c_2}(S))\\
  \sem{\IFTHENELSE{p = 0}{c_1}{c_2}}(S) &= \set{\sigma \in \sem{c_1}(S) \mid \sigma(p) = 0} \cup \set{\sigma \in \sem{c_2}(S) \mid \sigma(p) \ne 0}\\
  \sem{\SWHILE{p \bowtie 0}{c}}(S) &= \nu (\lambda X. \set{\sigma \in S \mid \sigma(p) \not\bowtie 0} \cup \set{\sigma \in \sem{c}(X) \mid \sigma(p) \bowtie 0}),
\end{array}
\]
where \({\bowtie} \in \set{=, \ne}\) and \(\nu F\) is the greatest fixed
point of \(F\).  Intuitively, \(\sigma \in \sem{c}(S)\) means that
executing \(c\) from \(\sigma\) results in a state in \(S\) if the
execution terminates; notice that \(\sigma\) should be in \(\sem{c}(S)\)
if \(c\) does not terminate.  The semantics uses the greatest fixed
point instead of the least fixed point in the \(\WHILE\) statement so
that \(\sem{c}(S)\) contains the states from which the execution of
\(c\) does not terminate.  If we used the least fixed point in the
semantics of a while loop, then only the initial states from which the
program terminates would be in the denotation of the loop.  For example,
consider the following program \(P\) that does not terminate for any
initial state: \(\SWHILE{0=0}{\SKIP}\).  Then, \(\sem{P}(S)\) should be
\(\STATES\).  However, if the denotation of a \(\WHILE\) loop were given
by the least fixed point, then \(\sem{P}(S)\) would be \(\emptyset\).

\begin{example}
 \label{ex:concrete} Recall the program $c_\FREEFALL$ in
 Figure~\ref{fig:freefall}.  Let $p_1$ be
 $-gt+gt_0-v+v_0-x\rho+x_0\rho$, $p_2$ be $-gt^2 +gt_0^2 - 2tv +2t_0v_0
 + 2x - 2x_0$, $p$ be $p_1 + p_2$, and $S$ be $\set{ \sigma \in \STATES
 | \sigma(p) = 0}$.  We show that \(\sem{c_\FREEFALL}(S) = \STATES\).
 We write \(c_1\) for \((x,v,t) := (x_0,v_0,t_0)\), and \(c_2\) for
 \((x,v,t) := (x+v\DT,v-g\DT-\rho v \DT,t+\DT)\).  We have
 \(\sem{c_\FREEFALL}(S) = \sem{c_1}(\sem{\SWHILE{t - a \ne 0}{c_2}}(S))
 = \sem{c_1}(\nu F)\) where $F(X) = \set{\sigma \in S | \sigma(t-a) = 0}
 \cup \set{\sigma \in \sem{c_2}(X) | \sigma(t-a) \ne 0}$.  It is easy to
 check that $\sem{c_1}(S) = \STATES$, so it suffices to show that $\nu F
 \supseteq S$.  This holds because \(S\) is a fixed point of \(F\).
 Indeed, $F(S) = \set{\sigma \in S | \sigma(t-a) = 0} \cup \set{\sigma
 \in \sem{c_2}(S) | \sigma(t-a) \ne 0} = \set{\sigma \in S | \sigma(t-a)
 = 0} \cup \set{\sigma \in S | \sigma(t-a) \ne 0} = S$ as desired.  Note
 that $\sem{c_2}(S) = S$ because $c_2$ does not change the value of $p$.
\end{example}


The abstract semantics is essentially the same as that given by Cachera
et al.~\cite{DBLP:journals/scp/CacheraJJK14} with a small adjustment.
The preorder
\(\vorder \subseteq \POWERSET(\POLY[x_1,\dots,x_n]) \times
\POWERSET(\POLY[x_1,\dots,x_n])\) is defined by \(S_1 \vorder S_2
:\iff S_2 \subseteq S_1\)\footnote{The original abstract semantics of
  Cachera et al.~\cite{DBLP:journals/scp/CacheraJJK14} is defined as a
  transformer on \emph{ideals} of polynomials; however, we formulate
  it here so that it operates on \emph{sets} of polynomials because
  their invariant-synthesis algorithm depends on the choice of a
  generator of an ideal.}.  Then $\POWERSET(\POLY[x_1,\dots,x_n])$ is
a complete lattice, and the meet is given as the set unions: Given $H
\in \POWERSET(\POLY[x_1,\dots,x_n])$ and $U \subseteq
\POWERSET(\POLY[x_1,\dots,x_n])$, $H \vorder G$ for all $G \in U$ if
and only if $H \vorder \bigcup U$.

The abstraction \(\alpha(S)\) is defined by \(\set{p \in
  \POLY[x_1,\dots,x_n] | \forall \sigma \in S, \sigma(p) = 0}\), the
  polynomials evaluated to \(0\) under all the states of \(S\).  The
  concretization \(\gamma(G)\) is defined by \(\set{\sigma \in \STATES |
  \forall p \in G, \sigma(p) = 0}\), the states that evaluate all the
  polynomials in \(G\) to \(0\).  The pair of \(\alpha\) and \(\gamma\)
  constitutes a Galois connection; indeed, both $\alpha(S) \vorder G$
  and $S \subseteq \gamma(G)$ are by definition equivalent to the
  following: $\forall p \in G, \forall \sigma \in S, \sigma(p) = 0$.
For example, the set of a state \(\set{\set{x_1 \mapsto 1, x_2 \mapsto
    0}}\) is abstracted by the set \(\set{(x_1 - 1)p_1 + x_2p_2 \mid
    p_1,p_2 \in \POLY[x_1,\dots,x_n]}\); this set is equivalently
    \(\Span{x_1-1,x_2}\), the ideal generated by \(x_1 - 1\) and
    \(x_2\).

The definition of the abstract semantics is parameterized over a
remainder-like operation \(\REM(f, p)\) that satisfies \(\REM(f,p) = f -
q p\) for some \(q\); we allow any $\REM$ that satisfies this condition
to be used.  Note that this differs from the standard remainder
operation where we require $\LM_{\preceq}(p)$ --- the greatest monomial
in $p$ with respect to a monomial order $\preceq$ --- not to divide any
monomial in $\LM_{\preceq}(\REM(f,p))$.  We write $\REM(G, p)$, where
\(G\) is a set of polynomials, for the set $\set{\REM(f,p) | f \in G
\setdiff \set{0}}$.


The abstract semantics \(\WPA{c}{\REM}\) is defined as follows.
\[
\footnotesize
\begin{array}{rl}
  \WPA{c}{\REM} &: (\POWERSET(\POLY[x_1,\dots,x_n]),\vorder) \ra (\POWERSET(\POLY[x_1,\dots,x_n]),\vorder)\\
  \WPA{\SKIP}{\REM}(G) & = G\\
  \WPA{\ASSIGN{x}{p}}{\REM}(G) & = G[x:=p] \\
  \WPA{\SEQ{c_1}{c_2}}{\REM}(G) &= \WPA{c_1}{\REM}(\WPA{c_2}{\REM}(G)) \\
  \WPA{\IFTHENELSE{p=0}{c_1}{c_2}}{\REM}(G)
                  &= p \cdot \WPA{c_2}{\REM}(G) \cup \REM(\WPA{c_1}{\REM}(G), p) \\
  \WPA{\SWHILE{p \ne 0}{c}}{\REM}(G)
                  &= \nu(\lambda H. p \cdot \WPA{c}{\REM}(H)
                    \cup \REM(G, p))\\
  \WPA{\SWHILE{p = 0}{c}}{\REM}(G)
                  &= \nu(\lambda H. p \cdot G
                    \cup \REM(\WPA{c}{\REM}(H), p)).
\end{array}
\]
In this definition,
$G[x:=p] = \set{q[x:=p] | q \in G}$ and $q[x:=p]$ is the polynomial
obtained by replacing $x$ with $p$ in $q$.  $\nu F$ exists for an arbitrary monotone
$F$ because we are working in the complete lattice
$\POWERSET(\POLY[x_1,\dots,x_n])$; concretely, we have $\nu F =
\bigcup \set{ G | G \vorder F(G)}$.

\(\WPA{c}{\REM}\) transfers backward a set of polynomials whose values
are \(0\).  Cachera et
al.~\cite[Theorem~3]{DBLP:journals/scp/CacheraJJK14} showed the
soundness of this abstract semantics: For any program \(c\) and a set of
polynomials \(G\), we have \(\gamma(\WPA{c}{\REM}(G)) \subseteq
\sem{c}(\gamma(G))\).  Although our abstract values are sets rather than
ideals, we can prove this theorem in the same way (i.e., induction on
the structure of $c$) as the original proof.

The highlight of the abstract semantics is the definition of
$\WPA{\IFTHENELSE{p=0}{c_1}{c_2}}{\REM}$.  In order to explain this
case, let us describe a part of the soundness proof: We show
$\gamma(\WPA{\IFTHENELSE{p=0}{c_1}{c_2}}{\REM}(G)) \subseteq
\sem{\IFTHENELSE{p=0}{c_1}{c_2}}(\gamma(G))$ assuming
$\gamma(\WPA{c_1}{\REM}(G)) \subseteq \sem{c_1}(\gamma(G))$ and
$\gamma(\WPA{c_2}{\REM}(G)) \subseteq \sem{c_2}(\gamma(G))$.  Suppose
$\sigma \in \gamma(\WPA{\IFTHENELSE{p=0}{c_1}{c_2}}{\REM}(G))$.  
Our goal is to show $\sigma \in
\sem{\IFTHENELSE{p=0}{c_1}{c_2}}(\gamma(G))$.
Therefore,
it suffices to show that (1) $\sigma(p) = 0$ implies $\sigma \in
\sem{c_1}(\gamma(G))$, and (2) $\sigma(p) \ne 0$ implies $\sigma \in
\sem{c_2}(\gamma(G))$.

\begin{itemize}
\item We first show that if $\sigma(p) = 0$ then
  $\sigma \in \sem{c_1}(\gamma(G))$.  By the induction hypothesis, we
      have
  $\gamma(\WPA{c_1}\REM(G)) \subseteq \sem{c_1}(\gamma(G))$, so it
  suffices to show that $\sigma \in \gamma(\WPA{c_1}\REM(G))$.  Take
  $f \in \WPA{c_1}{\REM}(G)$.  Then there exists
  $r \in \REM(\WPA{c_1}\REM(G),p)$ and $q \in \POLY[x_1,\dots,x_n]$ such
  that $f = qp + r$.  Because $\sigma(p)=0$ and $r \in \REM(\WPA{c_1}\REM(G),p) \subseteq \WPA{\IFTHENELSE{p=0}{c_1}{c_2}}{\REM}(G)$ and $\sigma \in \WPA{\IFTHENELSE{p=0}{c_1}{c_2}}{\REM}(G)$, we have
  $\sigma(f) = \sigma(q) \sigma(p) + \sigma(r) = 0$.  Since $f$ is an
  arbitrary element of $\WPA{c_1}{\REM}(G)$, by definition of $\gamma$
  we conclude that $\sigma \in \gamma(\WPA{c_1}\REM(G))$.
\item Next we show that $\sigma(p) \ne 0$ implies
  $\sigma \in \sem{c_2}(\gamma(G))$.
  By the induction hypothesis, we have $\gamma(\WPA{c_2}{\REM}(G))
      \subseteq
  \sem{c_2}(\gamma(G))$, so it suffices to show that $\sigma \in
  \gamma(\WPA{c_2}{\REM}(G))$.
  Take $f \in \WPA{c_2}{\REM}(G)$.  Then $pf \in p \cdot
  \WPA{c_2}{\REM}(G) \subseteq \WPA{\IFTHENELSE{p=0}{c_1}{c_2}}{\REM}(G)$, thus $\sigma(pf) = 0$.
  From the assumption $\sigma(p) \ne 0$, this implies $\sigma(f) = 0$.
  Since $f$ is arbitrary, we conclude that $\sigma
  \in \gamma(\WPA{c_2}\REM(G))$\footnote{The soundness would still hold
      even if we defined $\WPA{\IFTHENELSE{p=0}{c_1}{c_2}}{\REM}(G)$ by
      $\WPA{c_2}{\REM}(G) \cup \REM(\WPA{c_1}{\REM}(G), p)$ instead of
      $p \cdot \WPA{c_2}{\REM}(G) \cup \REM(\WPA{c_1}{\REM}(G), p)$.
      The multiplier $p$ makes the abstract semantics more precise.}.
\end{itemize}


The abstract semantics is related to the postcondition problem as
follows:
\begin{theorem}
  If \(\WPA{c}{\REM}(G) = \set{0}\), then \(\sem{c}(\gamma(G))
  = \STATES\) (hence \(g = 0\) is a solution of the postcondition
  problem for any \(g \in G\)).
  \label{thm:soundness-cachera-et-al}
\end{theorem}
\begin{proof}
  From the soundness above,
  \(\gamma(\WPA{c}{\REM}(G)) = \gamma(\set{0}) = \STATES \subseteq
  \sem{c}(\gamma(G))\); therefore \(\sem{c}(\gamma(G)) = \STATES\)
  follows because \(\STATES\) is the top element in the concrete
  domain.
\end{proof}

\begin{example}
 \label{ex:abstract} We exemplify how the abstract semantics works using
 the program \(c_\FREEFALL\) in Figure~\ref{fig:freefall}.  Set \(p\),
 \(c_1\), and \(c_2\) as in Example~\ref{ex:concrete}.  Define \(\REM\)
 in this example by \(\REM(f, p) = f\).  First, let $F(H) :=
 (t-a)\WPA{c_2}{\REM}(H) \cup \set{p}$, \(g_0 := 1\), \(g_{n+1} :=
 (t-a)(t+\DT-a)\dots(t+n\DT-a)\), and $G = \set{ g_np | n \in \NAT}$.
 Then $\nu F = G$.  Indeed, by definition of $\vorder$, we have $\top =
 \emptyset$, and it is easy to check that $F^n(\top) = \set{ g_kp | 0
 \le k < n}$.  Therefore we have $\nu F \vorder G$, because \(G\) is the
 greatest lower bound of \((F^n(\top))_{n \in \NAT}\).  By simple
 computation, we can see that $G$ is a fixed point of $F$, so we also
 have $G \vorder \nu F$; hence, $\nu F = G$.  Therefore,
 \(\WPA{c_\FREEFALL}{\REM}(\set{p}) = \set{0}\):
 \(\WPA{c_\FREEFALL}{\REM}(\set{p}) = \WPA{c_1}{\REM}(\WPA{\SWHILE{t - a
 \ne 0}{c_2}}{\REM}(\set{p})) = \WPA{c_1}{\REM}(\nu(\lambda
 H. (t-a)\WPA{c_2}{\REM}(H) \cup \REM(\set{p},t-a))) =
 \WPA{c_1}{\REM}(\nu(\lambda H. (t-a)\WPA{c_2}{\REM}(H) \cup \set{p})) =
 \WPA{c_1}{\REM}(\set{ g_np | n \in \NAT}) = \set{ (g_np)[x := x_0, v :=
 v_0, t := t_0] | n \in \NAT} = \set{0}\).
\end{example}

By Theorem~\ref{thm:soundness-cachera-et-al}, a set of polynomials $G$
such that $\WPA{c}{\REM}(G) = \set{0}$ for \emph{some} $\REM$
constitutes a solution of the postcondition problem.  The choice of
$\REM$ indeed matters in solving the postcondition problem: There are
$c$ and $G$ such that $\WPA{c}{\REM}(G) = \set{0}$ holds for some $\REM$
but not for others.  The reader is referred to~\cite[Section
4.1]{DBLP:journals/scp/CacheraJJK14} for a concrete example.



\section{Generalized homogeneous polynomials}
\label{sec:ghpoly}

\subsection{Definition}
\label{sec:ghpoly-def}

A polynomial \(p\) is said to be a homogeneous polynomial of degree
\(d\) if the degree of every monomial in \(p\) is
\(d\)~\cite{Cox:2007:IVA:1204670}.  As we mentioned in
Section~\ref{sec:introduction}, we generalize this notion of
homogeneity.

We first generalize the notion of the degree of a polynomial.
\begin{definition}
  \label{def:gdeg} The group of \emph{generalized degrees (g-degrees)}
  \(\GDEGS_B\), ranged over by \(\tau\), is an Abelian group freely
  generated by the finite set \(B\); that is, \(\GDEGS_B :=
  \set{b_1^{n_1} \dots b_m^{n_m} \mid b_1,\dots,b_m \in B, n_1,\dots,n_m
  \in \INT}\).  We call \(B\) the set of the \emph{base degrees}.  We
  often omit \(B\) in \(\GDEGS_B\) if the set of the base degrees is
  clear from the context.
\end{definition}
For example, if we set \(B\) to \(\set{L,T}\), then \(L, T\), and
\(LT^{-1}\) are all generalized degrees.  By definition, \(\GDEGS_B\)
has the multiplication on these g-degrees (e.g., \((LT) \cdot (LT^{-2})
= L^2T^{-1}\) and \((LT^2)^2 = L^2 T^4\)).

\begin{wrapfigure}{l}[0pt]{0.5\linewidth}
\begin{minipage}{\linewidth}
  \infax[T-Skip]{\Gamma \p \SKIP}
\end{minipage}\\
\vspace{1mm}
  \begin{minipage}{\linewidth}
    \infrule[T-Seq]{\Gamma \p c_1 \andalso \Gamma \p c_2}
            {\Gamma \p \SEQ{c_1}{c_2}}
  \end{minipage}\\
  \vspace{1mm}
\begin{minipage}{\linewidth}
  \infrule[T-Assign]{\Gamma(x) = \DDEG_\Gamma(p)}
          {\Gamma \p \ASSIGN{x}{p}}
\end{minipage}\\
\vspace{1mm}
  \begin{minipage}{\linewidth}
    \infrule[T-If]{\DDEG_\Gamma(p) = \tau \andalso \Gamma \p c_1 \andalso \Gamma \p c_2}
            {\Gamma \p \IFTHENELSE{p = 0}{c_1}{c_2}}
  \end{minipage}\\
  \vspace{1mm}
  \begin{minipage}{\linewidth}  
    \infrule[T-While]{\DDEG_\Gamma(p) = \tau \andalso \Gamma \p c}
            {\Gamma \p \SWHILE{p \bowtie 0}{c}}
  \end{minipage}
  \vspace{1mm}
  \caption{Typing rules}
  \label{fig:typingrules}
\end{wrapfigure}

In the analogy of quantity dimensions, the set \(B\) corresponds to
the base quantity dimensions (e.g., \(L\) for lengths and \(T\) for
times); the set \(\GDEGS_B\) corresponds to the derived quantity
dimensions (e.g., \(LT^{-1}\) for velocities and \(LT^{-2}\) for
acceleration rates.); multiplication expresses the relationship among
quantity dimensions (e.g., \(LT^{-1} \cdot T = L\) for
\(\mbox{velocity} \times \mbox{time} = \mbox{distance}\).)

\begin{definition}
  A \emph{g-degree assignment} is a finite mapping from \(\VARS\) to
  \(\GDEGS\).  A metavariable \(\Gamma\) ranges over the set of
  g-degree assignments.  For a power product \(w := x_1^{d_1} \dots
  x_n^{d_n}\), we write \(\DDEG_\Gamma(w)\) for \(\Gamma(x_1)^{d_1}
  \dots \Gamma(x_n)^{d_n}\) and call it the \emph{g-degree of \(w\)
    under \(\Gamma\)} (or simply \emph{g-degree} of \(w\) if
  \(\Gamma\) is not important); \(\DDEG_\Gamma(k w)\), the g-degree of
  a monomial \(k w\) under \(\Gamma\), is defined by
  \(\DDEG_\Gamma(w)\).
\end{definition}

For example, set \(\Gamma\) to \(\set{t \mapsto T, v \mapsto
  LT^{-1}}\); then \(\DDEG_\Gamma(2vt) = L\).  In terms of the analogy
with quantity dimensions, this means that the expression \(2vt\)
represents a length.

\begin{definition}
 We say \(p\) is a \emph{generalized homogeneous (GH) polynomial of
 g-degree \(\tau\) under \(\Gamma\)} if every monomial in \(p\) has the
 g-degree \(\tau\) under \(\Gamma\).  We write \(\DDEG_\Gamma(p)\) for
 the g-degree of \(p\) if it is a GH polynomial under \(\Gamma\); if it
 is not, then \(\DDEG_\Gamma(p)\) is not defined.  We write
 \(\POLY[x_1,\dots,x_n]_{\Gamma,\tau}\) for the set of the GH
 polynomials with g-degree \(\tau\) under \(\Gamma\).  We write
 \(\POLY[x_1,\dots,x_n]_{\Gamma}\) for \(\bigcup_{\tau \in \GDEGS}
 \POLY[x_1,\dots,x_n]_{\Gamma,\tau}\).
\end{definition}

\begin{example}
  \label{ex:gdegree}
  The polynomial \(-gt^2 +gt_0^2 - 2tv +2t_0v_0 + 2x - 2x_0\) (the
  polynomial \(p_2\) in Example~\ref{ex:concrete}) is a GH-polynomial
  under
  \[\Gamma := \set{g \mapsto LT^{-2}, t \mapsto T, v \mapsto
    LT^{-1}, x \mapsto L, x_0 \mapsto L, v_0 \mapsto LT^{-1}, \rho
    \mapsto T^{-1}, a \mapsto T}\] because all the monomials in
  \(p_2\) have the same g-degree in common; for example,
  \(\DDEG_{\Gamma}(-gt^2) = \Gamma(g)\Gamma(t)^2 = (LT^{-2})T^2 = L\);
  \(\DDEG_{\Gamma}(-2tv) = \Gamma(t)\Gamma(v) = T(LT^{-1}) = L\);
  \(\DDEG_{\Gamma}(2x) = \Gamma(x) = L\); and \(\DDEG_{\Gamma}(-2x_0)
  = \Gamma(x_0)= L\).  Therefore, \(\DDEG_{\Gamma}(p_2) = L\).  We
  also have \(\DDEG_{\Gamma}(p_1) = LT^{-1}\).
\end{example}


It is easy to see that any $p \in \POLY[x_1,\dots,x_n]$ can be uniquely
written as the finite sum of GH polynomials as \(p_{\Gamma,\tau_1} +
\dots + p_{\Gamma,\tau_m}\), where \(p_{\Gamma,\tau_i}\) is the summand
of g-degree \(\tau_i\) under \(\Gamma\) in this representation.  For
example, the polynomial \(p\) in Example~\ref{ex:concrete}, can be
written as \(p_{L} + p_{LT^{-1}}\) where \(p_{L} = p_1\) and
\(p_{LT^{-1}} = p_2\) from the previous example.  We call
\(p_{\Gamma,\tau}\) the \emph{homogeneous component of \(p\) with
g-degree \(\tau\) under \(\Gamma\)}, or simply \emph{a homogeneous
component of $p$}; we often omit \(\Gamma\) part if it is clear from the
context.



The definitions above are parameterized over a g-degree assignment
\(\Gamma\).  It is determined from the usage of variables in a given
program, which is captured by the following type judgment.

\begin{definition}
The judgment \(\Gamma \p c\) is the smallest relation that satisfies
the rules in Figure~\ref{fig:typingrules}.  We say \(\Gamma\) is
\emph{consistent} with the program \(c\) if \(\Gamma \p c\) holds.
\end{definition}
The consistency relation above is an adaptation of the \emph{dimension
  type system} proposed by
  Kennedy~\cite{DBLP:conf/esop/Kennedy94,kennedy96:_progr_languag_dimen}
  to our imperative language.  A g-degree assignment \(\Gamma\) such
  that \(\Gamma \p c\) holds makes every polynomial in \(c\) a GH one.
  In the rule \rn{T-Assign}, we require the polynomial \(p\) to have the
  same g-degree as that of \(x\) in \(\Gamma\).

\subsection{Automated inference of the g-degree assignment}
\label{sec:inference}

Kennedy also proposed a constraint-based automated type inference
algorithm of his type
system~\cite{DBLP:conf/esop/Kennedy94,kennedy96:_progr_languag_dimen}.
We adapt his algorithm so that, given a command $c$, it infers a
g-degree assignment $\Gamma$ such that $\Gamma \p c$.  The algorithm is
in three steps: (1) designating a template of the g-degree assignment,
(2) generating constraints over g-degrees, and (3) solving the
constraints.  In order to make the current paper self-contained, we
explain each step below.

\paragraph{Step 1: Designating a template of the g-degree assignment}

Let $S_c := \set{x_1,\dots,x_n}$ be the set of the variables occurring
in the given program $c$.  Then, the algorithm first designates a
template g-degree assignment $\Gamma_c := \set{x_1 \mapsto \alpha_{x_1},
\dots, x_n \mapsto \alpha_{x_n}}$ where
$\alpha_{x_1},\dots,\alpha_{x_n}$ are fresh unknowns taken from the set
$\GDEGV$ for the g-degrees of $x_1,\dots,x_n$.  For example, given the
program $c_\FREEFALL$ in Figure~\ref{fig:freefall}, the algorithm
designates
$$
\Gamma_{c_\FREEFALL} :=
\left\lbrace
  \begin{array}{l}
    g \mapsto \alpha_g, t \mapsto
    \alpha_t, \DT \mapsto \alpha_\DT, v \mapsto \alpha_v, x \mapsto
    \alpha_x, \\
    x_0 \mapsto \alpha_{x_0}, v_0 \mapsto \alpha_{v_0}, \rho
    \mapsto \alpha_{\rho}, a \mapsto \alpha_a
  \end{array}
\right\rbrace
$$
where $\alpha_g, \alpha_t,
\alpha_\DT, \alpha_v, \alpha_x, \alpha_{x_0}, \alpha_{v_0},
\alpha_{\rho}, \alpha_a$ are distinct unknowns for the g-degrees of the
variables that are to be inferred.

\paragraph{Step 2: Generating constraints over g-degrees}

The algorithm then generates the constraints over the g-degrees.  We
first define the set of constraints.  Let $\GDEGS'$ be
$\GDEGS_{\set{\alpha_1,\dots,\alpha_n}}$ in the rest of this section,
where $\alpha_1,\dots,\alpha_n$ are the unknowns generated in the
previous step.  (Recall that $\GDEGS_S$ is the set of g-degrees
generated by $S$.  Therefore, $\GDEGS'$ is the set of products of the
form $\alpha_1^{k_1} \dots \alpha_n^{k_n}$ for $k_1,\dots,k_n \in
\INT$.)  The $\Gamma_c$ generated in the previous step can be seen as a
map from $\VARS$ to $\GDEGS'$.

A \emph{g-degree constraint} is an equation $\tau_1 = \tau_2$ where
$\tau_1,\tau_2 \in \GDEGS'$.  We use a metavariable $\sigma$ for maps
from $\set{\alpha_1,\dots,\alpha_n}$ to $\GDEGS_B$.  This map can be
naturally extended to take the elements of $\GDEGS'$.  We say that
$\sigma$ is a \emph{solution} of a constraint set $C$ if it satisfies
all the equations in $C$.  For example, the map $\sigma := \set{\alpha_v
\mapsto LT^{-1}, \alpha_x \mapsto L, \alpha_t \mapsto T}$ is a solution
of the constraint set $\set{\alpha_v = \alpha_x\alpha_t^{-1}}$ since
$\sigma(\alpha_v) = LT^{-1} = \sigma(\alpha_x\alpha_t^{-1}) =
\sigma(\alpha_x)\sigma(\alpha_t)^{-1}$.


For a polynomial $p := a_1 w_1 + \dots + a_n w_n$, we write
$\DDEG'_{\Gamma}(p)$ for the pair $(\DDEG_{\Gamma}(w_1), C)$ where $C$
is $\emptyset$ if $n = 1$ and $\enumSet{\DDEG_{\Gamma}(w_1) =
\DDEG_{\Gamma}(w_2), \dots, \DDEG_{\Gamma}(w_{n-1}) =
\DDEG_{\Gamma}(w_n)}$ otherwise.  The intuition of $\DDEG'_{\Gamma}(p) =
(\tau, C)$ is that, for any solution $\sigma$ of $C$, the polynomial $p$
is generalized homogeneous and its g-degree is $\sigma(\tau)$.

For example, let $\Gamma$ be $\set{v \mapsto \alpha_v, g \mapsto
\alpha_g, x \mapsto \alpha_x}$ and $p$ be $2 v^2 + g x$; then,
$\DDEG'_{\Gamma}(p)$ is the pair $(\alpha_v^2, C)$ where $C$ is
$\set{\alpha_v^2 = \alpha_g \alpha_x}$.  For a solution $\sigma :=
\set{\alpha_v \mapsto LT^{-1}, \alpha_g \mapsto LT^{-2}, \alpha_x
\mapsto L}$ of $C$, $\sigma(\Gamma) = \enumSet{v \mapsto LT^{-1}, g \mapsto
LT^{-2}, x \mapsto L}$.  The polynomial $p$ is generalized homogeneous
under $\sigma(\Gamma)$ since $\sigma(\Gamma)(v^2) = \sigma(\Gamma)(g x)
= L^2T^{-2}$.  This is equal to $\sigma(\alpha_v^2)$.

The function $\PT$ for the constraint generation is defined as follows:
$$
\begin{array}{rcl}
 \PT(\Gamma,\SKIP)&:=&\emptyset\\
 \PT(\Gamma,\SEQ{c_1}{c_2})&:=&\PT(\Gamma,c_1) \cup \PT(\Gamma,c_2)\\
 \PT(\Gamma,\ASSIGN{x}{p})&:=& \set{\Gamma(x) = \tau} \cup C\\
 && \mbox{ where } (\tau,C) := \DDEG'_\Gamma(p)\\
 \PT(\Gamma,\IFTHENELSE{p = 0}{c_1}{c_2})&:=& C \cup \PT(\Gamma,c_1) \cup \PT(\Gamma,c_2)\\
 && \mbox{ where } (\tau,C) := \DDEG'_\Gamma(p)\\
 \PT(\Gamma,\SWHILE{p \bowtie 0}{c})&:=& \PT(\Gamma,c)\\
 && \mbox{ where } (\tau,C) := \DDEG'_\Gamma(p).\\
\end{array}
$$ 

The constraints $\PT(\Gamma,c)$ is defined so that its any solution
$\sigma$ satisfies $\sigma(\Gamma) \p c$.  The definition essentially
constructs the derivation tree of $\Gamma \p c$ following the rules in
Figure~\ref{fig:typingrules} and collects the constraints appearing in
the tree.

\begin{example}
 \label{ex:gDegreeConstraint} $\PT(\Gamma_{c_\FREEFALL},c_\FREEFALL)$
 generates the following constraints.  From the commands in Line 1, the
constraint set $\set{\alpha_x=\alpha_{x_0}, \alpha_v=\alpha_{v_0},
\alpha_t=\alpha_{t_0}}$ is generated; from the guard in Line 2,
$\set{\alpha_t = \alpha_a}$ is generated; from the right-hand side of
Line 3, the constraint set $\set{\alpha_x = \alpha_v \alpha_\DT,
\alpha_v = \alpha_g \alpha_\DT, \alpha_g\alpha_\DT =
\alpha_\rho\alpha_v\alpha_\DT, \alpha_t = \alpha_\DT}$, which ensures
the generalized homogeneity of each polynomial, is generated; $\PT$ also
generates $\set{\alpha_x = \alpha_x, \alpha_v = \alpha_v, \alpha_t =
\alpha_t}$, which ensures that the g-degrees of the left-hand side and
the right-hand side are identical.
\end{example}


\paragraph{Step 3: Solving the constraints}

The algorithm then calculates a solution of the generated constraints.
The constraint-solving procedure is almost the same as that by
Kennedy~\cite[Section 5.2]{DBLP:conf/esop/Kennedy94}, which is based on
Lankford's unification algorithm\footnote{We do not discuss the
termination of the procedure in this paper.  See Kennedy~\cite[Section
5.2]{DBLP:conf/esop/Kennedy94}.}~\cite{lankford1984abelian}.

The procedure obtains a solution $\sigma$ from the given constraint set
$C$ by applying the following rewriting rules successively: $$
\begin{array}{rcl}
 (\emptyset, \sigma) &\ra& \sigma\\
 (\set{\alpha'^{k}\vec{\alpha}^{\vec{n}} = 1} \cup C, \sigma) &\ra& (\set{\alpha' \mapsto \vec{\alpha}^{-\frac{\vec{n}}{k}}}(C), \set{\alpha' \mapsto \vec{\alpha}^{-\frac{\vec{n}}{k}}} \circ \sigma)\\
 && \mbox{where the absolute value of $k$ is not more than those of $\vec{n}$}\\
 && \mbox{(if $k$ divides all the integers in $\vec{n}$)}\\
 (\set{\alpha'^{k}\vec{\alpha}^{\vec{n}} = 1} \cup C, \sigma) &\ra& (\set{\omega^k \vec{\alpha}^{\vec{n} \MOD k} = 1} \cup \sigma'(C), \sigma' \circ \sigma)\\
 && \mbox{where the absolute value of $k$ is not more than those of $\vec{n}$,}\\
 && \mbox{$\sigma' = \set{\alpha' \mapsto \omega \vec{\alpha}^{-\floor{\frac{\vec{n}}{k}}}}$,}\\
 && \mbox{and $\omega$ is a fresh element of $\GDEGV$}\\
 && \mbox{(if there is an integer in $\vec{n}$ that is not divisible by $k$)}\\
 (\set{1 = 1} \cup C, \sigma) &\ra& (C, \sigma)\\
 (\set{\tau_1 = \tau_2} \cup C, \sigma) &\ra& (\set{\tau_1\tau_2^{-1} = 1} \cup C, \sigma)\\
 C &\ra& (C, \emptyset).\\
\end{array}
$$ 

The idea of the procedure is to construct a solution iteratively
converting a constraint $\alpha'^{k} \vec{\alpha}^{\vec{n}} = 1$ to
$\set{\alpha' \mapsto \vec{\alpha}^{-\frac{\vec{n}}{k}}}$ if $k$ divides
all the integers in $\vec{n}$ (i.e., the second case).  If $k$ does not
(i.e., the third case)\footnote{We do not use this case in the rest of
this paper.}, the procedure (1) splits $\frac{\vec{n}}{k}$ to the
quotient $\floor{\frac{\vec{n}}{k}}$ and the remainder $\vec{n} \MOD k$,
(2) generates a fresh g-degree variable $\omega$ representing
$\vec{\alpha}^{-\frac{\vec{n} \MOD k}{k}}$, and (3) sets $\alpha'$ in
the solution to $\omega \vec{\alpha}^{-\floor{\frac{\vec{n}}{k}}}$ which
is equal to $\vec{\alpha}^{-\frac{\vec{n}}{k}}$.

After obtaining a solution with the procedure above, the inference
algorithm assigns different base degree to each surviving g-degree
variable.




\begin{example}
 Consider the following constraint set $C$:
 $$
 \begin{array}{l}
  \left\{
   \begin{array}{l}
    \alpha_x\alpha_{x_0}^{-1}=1, 
     \alpha_v\alpha_{v_0}^{-1}=1,
     \alpha_t\alpha_{t_0}^{-1}=1,
     \alpha_t\alpha_\DT^{-1} = 1,\\
     \alpha_x\alpha_v^{-1}\alpha_\DT^{-1}=1, 
     \alpha_v \alpha_g^{-1} \alpha_\DT^{-1} = 1,
    \alpha_g\alpha_\DT \alpha_\rho^{-1} \alpha_v^{-1} \alpha_\DT^{-1} = 1
   \end{array}
  \right\}
 \end{array}
 $$ which is equivalent to that of Example~\ref{ex:gDegreeConstraint}.
 After several steps of rewriting, the procedure obtains
 $$
 \left(
 \left\{
 \begin{array}{l}
   \alpha_{v_0}\alpha_{g}^{-1}\alpha_{\DT}^{-1} = 1,\\
   \alpha_g \alpha_\rho^{-1} \alpha_{v_0}^{-1} = 1
 \end{array}
 \right\},
 \left\{
 \begin{array}{l}
  \alpha_x \mapsto \alpha_{v_0} \alpha_{\DT},
   \alpha_v \mapsto \alpha_{v_0},\\
   \alpha_t \mapsto \alpha_{\DT},
   \alpha_{t_0} \mapsto \alpha_{\DT},\\
   \alpha_{x_0} \mapsto \alpha_{v_0} \alpha_{\DT}
 \end{array}
 \right\} \right).  $$ 
 At the next step, suppose that the procedure
 picks up the constraint $\alpha_{v_0}\alpha_{g}^{-1}\alpha_{\DT}^{-1} =
 1$.  By applying the second rule, the procedure generates the following state
 $$
 \left(
 \left\{
 \begin{array}{l}
  \alpha_\rho^{-1} \alpha_{\DT}^{-1} = 1
 \end{array}
 \right\},
 \left\{
 \begin{array}{l}
  \alpha_x \mapsto \alpha_{g} \alpha_{\DT}^2,
   \alpha_v \mapsto \alpha_{g}\alpha_{\DT},\\
   \alpha_t \mapsto \alpha_{\DT},
   \alpha_{t_0} \mapsto \alpha_{\DT},\\
   \alpha_{x_0} \mapsto \alpha_{g}\alpha_{\DT}^2,
    \alpha_{v_0} \mapsto \alpha_{g}\alpha_{\DT}
 \end{array}
 \right\} \right).  
 $$ 

 Then, with the second and last rules, the procedure obtains the
 following solution:

 $$
 \left\{
 \begin{array}{l}
  \alpha_x \mapsto \alpha_{g}\alpha_{\DT}^2,
   \alpha_v \mapsto \alpha_{g}\alpha_{\DT},\\
   \alpha_t \mapsto \alpha_{\DT},
   \alpha_{t_0} \mapsto \alpha_{\DT},\\
   \alpha_{x_0} \mapsto \alpha_{g}\alpha_{\DT}^2,
    \alpha_{v_0} \mapsto \alpha_{g}\alpha_{\DT},
    \alpha_{\rho} \mapsto \alpha_{\DT}^{-1}
 \end{array}
 \right\}.  
 $$

 By assigning the base degree $A$ to $\alpha_g$ and $T$ to
 $\alpha_{\DT}$, we have the following solution:

 $$
 \left\{
 \begin{array}{l}
  \alpha_x \mapsto A T^2,
   \alpha_v \mapsto A T,
   \alpha_t \mapsto T,
   \alpha_{t_0} \mapsto T,\\
   \alpha_{x_0} \mapsto A T^2,
    \alpha_{v_0} \mapsto A T,
    \alpha_{\rho} \mapsto T^{-1}
 \end{array}
 \right\}.  
 $$

 Notice the set of base degrees is different from that we used in
 Example~\ref{ex:gdegree}; in this example, the g-degree for the
 acceleration rates ($A$) is used as a base degree, whereas that for
 lengths ($L$) is used in Example~\ref{ex:gdegree}.  This happens
 because the order of the constraints chosen in an execution of the
 inference algorithm is nondeterministic.  Our results in the rest of
 this paper do not depend on a specific choice of base degrees.

 \paragraph{Limitation}
 
 A limitation of the current g-degree inference algorithm is that, even
 if a constant symbol in a program is intended to be of a g-degree other
 than $1$, it has to be of g-degree $1$ in the current type system.  For
 example, consider the program $c_\FREEFALL'$ obtained by replacing $g$
 in $c_\FREEFALL$ with $9.81$ and $\rho$ with $0.24$.  Then, the
 g-degrees of $v$ and $\DT$ are inferred to be $1$ due to the assignment
 $v := v - 9.8 \DT - 0.24 v \DT$ in $c_\FREEFALL'$: The constraints for
 this assignment generated by the inference algorithm is $\set{\alpha_v
 = \alpha_{\DT}, \alpha_{\DT} = \alpha_v \alpha_\DT, \alpha_v =
 \alpha_v}$, whose only solution is $\set{\alpha_v \mapsto 1,
 \alpha_{\DT} \mapsto 1}$.  This degenerated g-degrees are propagated to
 the other variables during the inference of $c_\FREEFALL'$, leading to
 the g-degree assignment in which all the variables have the g-degree
 $1$.  This g-degree assignment is not useful for the template-size
 reduction; any polynomial is a GH polynomial under this assignment.

 As a workaround, our current implementation that will be described in
 Section~\ref{sec:experiment} uses an extension that can assign a
 g-degree other than $1$ to each occurrence of a constant symbol by
 treating a constant symbol as a variable.  For example, for the
 following program \(\mathtt{sumpower}_d\): \((x,y,s) :=
 (X+1,0,\underline{1}); \WHILE x \ne 0 \DO \IF y = 0 \THEN (x, y) := (x
 - 1, x) \ELSE (s, y) := (s + y^d, y - 1)\), the inference algorithm
 treats the underlined occurrence of $1$ as a variable and assigns
 \(T^d\) to it; the other occurrences of $0$ and $1$ are given g-degree
 \(T\).  This g-degree assignment indeed produces a smaller template.

\end{example}

\section{Abstract semantics restricted to GH polynomials}
\label{sec:ghsemantics}

This section gives the main result of this paper: If there is an
algebraic invariant of \(c\) and \(\Gamma \p c\), then there exists an
algebraic invariant that consists of a GH polynomial under \(\Gamma\).

To state this result formally, we revise our abstract semantics by
restricting it to the domain of the GH polynomials.  The domain is
obtained by replacing the underlying set of the domain
\(\POWERSET(\POLY[x_1,\dots,x_n])\) with
\(\POWERSET(\POLY[x_1,\dots,x_n]_\Gamma)\).  This is a subset of
\(\POWERSET(\POLY[x_1,\dots,x_n])\) that is closed under arbitrary
meets.  We can define the abstraction and the concretization in the
same way as in Section~\ref{sec:language}.

The revised abstract semantics \(\WPAH{c}{\REM,\Gamma}\), which we
hereafter call \emph{GH abstract semantics}, is the same as the
original one except that it is parameterized over the g-degree
assignment \(\Gamma\).  In the following definition, we write
\(\REM(G,p)\) for \(\set{\REM(f, p)
  \mid f \in (G \cap \POLY[x_1,\dots,x_n]_\Gamma) \backslash
  \set{0}}\), the set of the remainder obtained from a GH
polynomial in \(G\) and \(p\).  We assume that our choice of $\REM$ is a remainder
operation such that whenever both $f$ and $p$ are GH polynomials, so
is $\REM(f, p)$.
\[
\footnotesize
\begin{array}{rl}
  \WPAH{\SKIP}{\REM,\Gamma}(G) & = G\\
  \WPAH{\ASSIGN{x}{p}}{\REM,\Gamma}(G) & = G[x:=p] \\
  \WPAH{\SEQ{c_1}{c_2}}{\REM,\Gamma}(G) &= \WPAH{c_1}{\REM,\Gamma}(\WPAH{c_2}{\REM,\Gamma}(G)) \\
  \WPAH{\IFTHENELSE{p=0}{c_1}{c_2}}{\REM,\Gamma}(G)
                  &= p \cdot \WPAH{c_2}{\REM,\Gamma}(G) \cup \REM(\WPAH{c_1}{\REM,\Gamma}(G), p) \\
  \WPAH{\SWHILE{p \ne 0}{c}}{\REM,\Gamma}(G)
                  &= \nu(\lambda H. p \cdot \WPAH{c}{\REM,\Gamma}(H)
                    \cup \REM(G, p))\\
  \WPAH{\SWHILE{p = 0}{c}}{\REM,\Gamma}(G)
                  &= \nu(\lambda H. p \cdot G
  \cup \REM(\WPAH{c}{\REM,\Gamma}(H), p)).
\end{array}
\]

The following theorem guarantees that the invariant found using the
semantics \(\WPAH{c}{\REM,\Gamma}\) is indeed an invariant of \(c\).
\begin{theorem}[Soundness of the GH abstract semantics]
  \label{thm:soundness}

  If \(\Gamma \p c\) and \(G\) is a set of GH polynomials under
  \(\Gamma\), then \(\WPAH{c}{\REM,\Gamma}(G) =
  \WPA{c}{\REM}(G)\).
\end{theorem}
\begin{proof}
  By induction on \(c\).
\end{proof}
This theorem implies that if $g$ is a GH polynomial under \(\Gamma\)
and $\WPAH{c}{\REM,\Gamma}(g) = \set{0}$, then $g$ is indeed a
solution of the postcondition problem.



Completeness of \(\WPAH{c}{\REM,\Gamma}\) is obtained as a corollary of the
following lemma.
\begin{lemma}
  \label{lem:homogeneityPreserved}
  Suppose \(\Gamma \p c\), \(g_1',\dots,g_m' \in
  \POLY[x_1,\dots,x_n]\), and \(g_i\) is a
  homogeneous component of \(g_i'\) (i.e., \(g_i = {g_i'}_{\tau_i}\)
  for some \(\tau_i\)).  If \(h \in
  \WPAH{c}{\REM,\Gamma}(\set{g_1,\dots,g_m})\), then there exists \(h'
  \in \WPA{c}{\REM}(\set{g_1',\dots,g_m'})\) such that \(h\) is a
  homogeneous component of \(h'\).
\end{lemma}
\begin{proof}
  Let us say \(G\) is a homogeneous component of \(G'\) under \(\Gamma\)
  if, for any \(p \in G\), there exists \(p' \in G'\) such that \(p =
  p'_\tau\) for some \(\tau\).  By induction on $c$, we can prove that
  if \(G\) is a homogeneous component of \(G'\) under \(\Gamma\), then
  $\WPAH{c}{\REM,\Gamma}(G)$ is a homogeneous component of
  $\WPAH{c}{\REM,\Gamma}(G')$ under \(\Gamma\).
\end{proof}

\begin{theorem}[Completeness]
  \label{cor:completeness}
  Let $g_i$ and $g_i'$ be the same as in
  Lemma~\ref{lem:homogeneityPreserved}.  If \(\Gamma \p c\) and
  $\WPA{c}{\REM}(\set{g_1', \dots, g_m'}) = \set{0}$, then
  $\WPAH{c}{\REM,\Gamma}(\set{g_1, \dots, g_m}) = \set{0}$.
\end{theorem}
\begin{proof}
  Take $h \in \WPAH{c}{\REM,\Gamma} (\set{g_1, \dots, g_m})$.  Then
  there exists $h' \in \WPA{c}{\REM}(\enumSet{g_1', \dots, g_m'})$ such
  that $h'_{\DDEG(h)} = h$.  By assumption, we have $h' = 0$; therefore
  $h = 0$.
\end{proof}
Hence, if $g = 0$ is a solution of the postcondition problem,
then so is $g' = 0$ for every homogeneous component $g'$ of $g$.


\begin{example}
 \label{ex:wpah} Recall Example~\ref{ex:abstract}.
 Theorem~\ref{cor:completeness} and $\WPA{c_\FREEFALL}{\REM}(\set{p}) =
 \set{0}$ guarantee \(\WPAH{c_\FREEFALL}{\REM,\Gamma}(\set{p_1}) =
 \set{0}\) and \(\WPAH{c_\FREEFALL}{\REM,\Gamma}(\set{p_2}) = \set{0}\)
 since $p_1$ and $p_2$ are homogeneous components of $p$.
\end{example}

\section{Template-based algorithm}
\label{sec:template}

This section applies our idea to Cachera's template-based
invariant-synthesis algorithm~\cite{DBLP:journals/scp/CacheraJJK14}.  We
hereafter use metavariable \(a\) for a \emph{parameter} that represents
an unknown value.  We use metavariable \(A\) for a set of parameters.  A
\emph{template} on $A$ is an expression of the form \(a_1 p_1 + \dots +
a_n p_n\) where \(a_1, \dots, a_n \in A\); we abuse the metavariable
\(G\) for a set of templates.  We denote the set of templates on \(A\)
by \(T(A)\).  A \emph{valuation} \(v\) on \(A\) is a map from \(A\) to
\(K\).  We can regard $v$ as a map from $T(A)$ to $\POLY[x_1, \dots,
x_n]$ by $v(a_1 p_1 + \dots + a_m p_m) = v(a_1) p_1 + \dots + v(a_m)
p_m$.

\subsection{Algorithm proposed by Cachera et al.}
\label{sec:theirMethod}


Cachera et al. proposed a sound template-based algorithm for the
postcondition problem.  Their basic idea is to express a fixed point by
constraints on the parameters in a template in order to avoid
fixed-point iteration.

To recall the algorithm of Cachera et al., we establish several
definitions.
\begin{definition}
An \emph{equality constraint} on $A$ is an expression of the form
$\EQCONSTR{G}{G'}$, where $G, G' \subseteq T(A)$.  A \emph{constraint
set} on \(A\), or simply \emph{constraints}, is a set of equality
constraints on $A$; a constraint set is represented by the metavariable
\(C\).  We may write \((A,C)\) for a constraint set $C$ on $A$ to make
$A$ explicit.  A valuation $v$ on $A$ \emph{satisfies} an equality
constraint $\EQCONSTR{G}{G'}$ on $A$, written \(v \models
\EQCONSTR{G}{G'}\), if $v(G)$ and $v(G')$ generate the same ideal.  A
\emph{solution} of a constraint set $(A, C)$ is a valuation on $A$ that
satisfies all constraints in $C$.  If $v$ is a solution of $(A, C)$, we
write $v \models (A, C)$, or simply $v \models C$.  A template $a_1 p_1
+ \dots + a_m p_m$ is a \emph{GH template of g-degree $\tau$ under
$\Gamma$} if \(p_1,\dots,p_m\) are GH polynomials of g-degree \(\tau\).
\end{definition}

We extend the definition of the remainder computation to operate on
templates.
\begin{definition}
  \label{def:rem-par}
\(\REMP(A, f, p)\) is a pair $(A', f - pq)$ where $q$ is the most
general template of degree $\PDEG(f) - \PDEG(p)$, the parameters of
which are fresh; $A'$ is the set of the parameters appearing in $q$.
We write \(\REMP(A, \set{p_1,\dots,p_m}, p)\) for \((A',G')\), where
\((A_i, r_i) = \REMP(A, p_i, p)\) and \(A' = \bigcup A_i\) and \(G' =
\set{r_1,\dots,r_m}\).
\end{definition}
For example, if the set of variables is \(\set{x}\), then
\(\REMP(\emptyset, x^2, x + 1) = (\set{a_1,a_2}, x^2 - (a_1 x + a_2)(x
+ 1))\); the most general template of degree \(\PDEG(x^2) - \PDEG(x +
1) = 1\) with variable \(x\) is \(a_1 x + a_2\).  By expressing a
remainder using a template, we can postpone the choice of a remainder
operator to a later stage; for example, if we instantiate
\((a_1,a_2)\) with \((1,-1)\), then we have the standard remainder
operator on \(\REAL[x]\).

We recall the constraint generation algorithm proposed by Cachera et
al.  We write \((A_i, G_i,C_i)\) for \(\WPC{c_i}{\REMP}(A,G,C)\) in
each case of the following definition.
\[
\footnotesize
\begin{array}{rl}
  \WPC{\SKIP}{\REMP}(A, G, C) & = (A, G, C)\\
  \WPC{\ASSIGN{x}{p}}{\REMP}(A, G, C) & = (A, G[x:=p], C) \\
  \WPC{\SEQ{c_1}{c_2}}{\REMP}(A, G, C) &= \WPC{c_1}{\REMP}(\WPC{c_2}{\REMP}(A, G, C)) \\
  \WPC{\IFTHENELSE{p=0}{c_1}{c_2}}{\REMP}(A, G, C)
                     &= (A_3, p \cdot G_2 \cup G_3, C_1 \cup C_2)\\
  \mbox{where } & \begin{array}[t]{rcl}
                (A_3, G_3) &=& \REMP(A_1 \cup A_2, G_1, p)\\
                \end{array}\\
  \WPC{\SWHILE{p \bowtie 0}{c_1}}{\REMP} (A, G, C)
                     &= (A_1, G, C_1 \cup \{\EQCONSTR{G}{G_1}\})\\
\end{array}
\]
\(\WPC{c}{\REMP}(A, G, C)\) accumulates the generated parameters to
\(A\) and the generated constraints to \(C\).  \(A\) is augmented by
fresh parameters at the \(\IF\) statement where \(\REMP\) is called.
At a \(\WHILE\) statement, \(\EQCONSTR{G}{G_1}\) is added to the
constraint set to express the loop-invariant condition.

\begin{algorithm}[t]
  \footnotesize
  \caption{Inference of polynomial invariants.}
  \label{alg:invariant-inference}
  \begin{algorithmic}[1]
    \Procedure{$\INVINF$}{$c$, $d$}
    \State $g \gets \text{the most general template of degree $d$}$
    \State $A_0 \gets \text{the set of the parameters occurring in $g$}$
    \State $(A, G, C) \gets \WPC{c}{\REMP}(A_0, \set{g}, \emptyset)$
    \State \Return $v(g)$ where $v$ is a solution of $C \cup \set{
      \EQCONSTR{G}{\set{0}}}$
    \EndProcedure
  \end{algorithmic}
\end{algorithm}

Algorithm~\ref{alg:invariant-inference} solves the postcondition
problem with the constraint-generating subprocedure
\(\WPC{c}{\REMP}\).  This algorithm, given a program \(c\) and degree
\(d\), returns a set of postconditions that can be expressed by an
algebraic condition with degree \(d\) or lower.  The algorithm
generates the most general template \(g\) of degree \(d\) for the
postcondition and applies \(\WPC{c}{\REMP}\) to \(g\).  For the
returned set of polynomials \(G\) and the constraint set \(C\), the
algorithm computes a solution of \(C \cup \EQCONSTR{G}{\set{0}}\); the
equality constraint \(\EQCONSTR{G}{\set{0}}\) states that \(v(g) =
0\), where \(v\) is a solution of the constraint set \(C \cup
\EQCONSTR{G}{\set{0}}\), has to hold at the end of \(c\) regardless of
the initial state.

This algorithm is proved to be sound: If \(p \in \INVINF(c,d)\), then
\(p = 0\) holds at the end of \(c\) for any initial
states~\cite{DBLP:journals/scp/CacheraJJK14}.  Completeness was not
mentioned in their paper.

\begin{remark}
The algorithm requires a solver for the constraints of the form
\(\EQCONSTR{G}{G'}\).  This is the problem of finding \(v\) that equates
\(\Span{G}\) and \(\Span{G'}\); therefore, it can be solved using a
solver for the ideal membership problems~\cite{Cox:2007:IVA:1204670}.
To avoid high-cost computation, Cachera et al. proposed heuristics to
solve an equality constraint.
\end{remark}

\begin{example}
  \label{ex:template}
  We explain how \(\INVINF(c_\FREEFALL,3)\) works.  The algorithm
  generates a degree-3 template \(q(x,v,t,x_0,v_0,t_0,a,\DT,g,\rho)\)
  over \(\set{x,v,t,x_0,v_0,t_0,a,\DT,g,\rho}\).  The algorithm then
  generates the following constraints by
  \(\WPCH{c_\FREEFALL}{\REMP}\):
  \(\langle\enumSet{q(x,v,t,x_0,v_0,t_0,a,\DT,g,\rho)} \equiv \enumSet{q(x +
    v\DT,v - g \DT - \rho v
    \DT,t+\DT,x_0,v_0,t_0,a,\DT,g,\rho)}\rangle\) (from the body of
  the loop) and
  \(\EQCONSTR{\enumSet{q(x_0,v_0,t_0,x_0,v_0,t_0,a,\DT,g,\rho)}}{\enumSet{0}}\).
 By solving these constraints with a solver for ideal membership
 problems~\cite{Cox:2007:IVA:1204670} or with the heuristics proposed by
 Cachera et al.~\cite{DBLP:journals/scp/CacheraJJK14}, and by applying
 the solution to $q(x,v,t,x_0,v_0,t_0,a,\DT,g,\rho)$, we obtain \(p\) in Example~\ref{ex:concrete}.
\end{example}

\subsection{Restriction to GH templates}

We define a variation \(\WPCH{c}{\REMPH{\Gamma},\Gamma}\) of the constraint
generation algorithm in which we use only GH polynomial templates.
\(\WPCH{c}{\REMPH{\Gamma},\Gamma}\) differs from \(\WPC{c}{\REMP}\) in that it
is parameterized also over \(\Gamma\), not only over the remainder
operation used in the algorithm.  The remainder operator
\(\REMPH{\Gamma}H{\Gamma}(A, f, p)\) returns a pair $(A \cup A', f - pq)$ where
$q$ is the most general GH template with g-degree
$\DDEG(f)\DDEG(p)^{-1}$, with degree $\PDEG(f) - \PDEG(p)$, and with
fresh parameters; \(A'\) is the set of the parameters that appear in
\(q\).  \(\REMPH{\Gamma}(A, G, p)\) is defined in the
same way as Definition~\ref{def:rem-par} for a set \(G\) of polynomials.
We again write \((A_i, G_i,C_i)\) for
\(\WPC{c_i}{\REMPH{\Gamma}}(A,G,C)\) in each case of the following definition.
\[
\footnotesize
\begin{array}{rl}
  \WPCH{\SKIP}{\REMPH{\Gamma},\Gamma}(A, G, C) & = (A, G, C)\\
  \WPCH{\ASSIGN{x}{p}}{\REMPH{\Gamma},\Gamma}(A, G, C) & = (A, G[x:=p], C) \\
  \WPCH{\SEQ{c_1}{c_2}}{\REMPH{\Gamma},\Gamma}(A, G, C) &= \WPCH{c_1}{\REMPH{\Gamma},\Gamma}(\WPCH{c_2}{\REMPH{\Gamma},\Gamma}(A, G, C)) \\
  \WPCH{\IFTHENELSE{p=0}{c_1}{c_2}}{\REMPH{\Gamma},\Gamma} (A, G, C)
                     &= (A_3, p \cdot G_2 \cup G_3, C_1 \cup C_2)\\
  &\hspace{-1cm}\text{where }
  \begin{array}[t]{ll}
    (A_3, G_3) &= \REMPH{\Gamma}H{\Gamma}(A_1 \cup A_2, G_1, p) \\
  \end{array}\\
  \WPCH{\SWHILE{p \bowtie 0}{c_1}}{\REMPH{\Gamma},\Gamma} (A, G, C)
  &= (A_1, G, C_1 \cup \{\EQCONSTR{G}{G_1}\})
\end{array}
\]

Algorithm~\ref{alg:invariant-inference-gh} is a variant of
Algorithm~\ref{alg:invariant-inference}, in which we restrict a template
to GH one.

\begin{algorithm}[t]
  \footnotesize \caption{Inference of polynomial invariants (homogeneous
  version).}  \label{alg:invariant-inference-gh}
  \begin{algorithmic}[1]
    \Procedure{$\INVINFH$}{$c$, $d$, \(\Gamma\), $\tau$}
    \State $g \gets \text{the most general template of g-degree
      $\tau$ and degree $d$}$
    \State $A_0 \gets \text{the set of the parameters occurring in $g$}$
    \State $(A, G, C) \gets \WPCH{c}{\REMPH{\Gamma},\Gamma}(A_0, \set{g}, \emptyset)$
    \State \Return $v(g)$  where $v$ is a solution of
    $C \cup \set{\EQCONSTR{G}{\set{0}}}$
    \EndProcedure
  \end{algorithmic}
\end{algorithm}

The algorithm \(\INVINFH\) takes the input \(\tau\) that specifies the
g-degree of the invariant at the end of the program \(c\).  We have not
obtained a theoretical result for \(\tau\) to be passed to \(\INVINFH\)
so that it generates a good invariant.  However, during the experiments
in Section~\ref{sec:experiment}, we found that the following strategy
often works: \emph{Pass the g-degree of the monomial of interest}.  For
example, if we are interested in a property related to \(x\), then pass
\(\Gamma(x)\) (i.e., \(L\)) to \(\INVINFH\) for the invariant \(-gt^2
+gt_0^2 - 2tv + 2t_0v_0 + 2x - 2x_0 = 0\).  How to help a user to find
such ``monomial of her interest'' is left as an interesting future
direction.

The revised version of the invariant inference algorithm is sound; at
the point of writing, completeness of \(\INVINFH\) with respect to
\(\INVINF\) is open despite the completeness of
\(\WPAH{c}{\REM,\Gamma}\) with respect to \(\WPA{c}{\REM}\).
\begin{theorem}[Soundness]
  \label{thm:hconstr-constr-sound} Suppose $\Gamma \p c$,
  $d \in \NAT$, and $\tau \in \GDEGS$.  Set \(P_1\) to the set of
  polynomials that can be returned by \(\INVINFH(c, d, \tau)\); set
  \(P_2\) to those by \(\INVINF(c, d)\).  Then, \(P_1 \subseteq P_2\).
\end{theorem}


\section{Experiment}
\label{sec:experiment}

We implemented Algorithm~\ref{alg:invariant-inference-gh} and
conducted experiments.  Our implementation \(\mathrm{Fastind}_{dim}\)
takes a program \(c\), a maximum degree \(d\) of the template \(g\) in
the algorithm, and a monomial \(w\).  It conducts type inference of
\(c\) to generate \(\Gamma\) and calls
\(\INVINFH(c,d,\Gamma,\DDEG_\Gamma(w))\).  The type inference
algorithm is implemented with OCaml; the other parts (e.g., a solver
for ideal-equality constraints) are implemented with Mathematica.

\newcommand\SIGDIGIT[2]{#1 \ifnum#2=0 \else \times 10\ifnum#2=1 \else ^{#2} \fi \fi}
\newcommand\SD[2]{\(\SIGDIGIT{#1}{#2}\)}

To demonstrate the merit of our approach, we applied this implementation
to the benchmark used in the experiment by Cachera et
al.~\cite{DBLP:journals/scp/CacheraJJK14} and compared our result with
that of their implementation, which is called Fastind.  The entire
experiment was conducted on a MacBook Air 13-inch Mid 2013 model with a
1.7 GHz Intel Core i7 (with two cores, each of which has 256 KB of L2
cache) and 8 GB of RAM (1600 MHz DDR3).  The modules written in OCaml
were compiled with \verb|ocamlopt|.  The version of OCaml is 4.02.1.
The version of Mathematica is 10.0.1.0.  We refer the reader
to~\cite{benchmarksPolynomialInvariants,DBLP:journals/scp/CacheraJJK14,DBLP:journals/jsc/Rodriguez-CarbonellK07}
for detailed descriptions of each program in the benchmark.  Each
program contains a nested loop with a conditional branch (e.g.,
\texttt{dijkstra}), a sequential composition of loops (e.g.,
\texttt{divbin}), and nonlinear expressions (e.g., \texttt{petter(n)}.)
We generated a nonlinear invariant in each program.

Table~\ref{table:result} shows the result.  The column deg shows the
degree of the generated polynomial, \(t_{\mathit{sol}}\) shows the time
spent by the ideal-equality solver (ms), \(\# m\) shows the number of
monomials in the generated template, \(t_{\mathit{inf}}\) shows the time
spent by the dimension-type inference algorithm (ms), and
\(t_{\mathit{inf}}+t_{\mathit{sol}}\) shows the sum of
\(t_{\mathit{inf}}\) and \(t_{\mathit{sol}}\).  By comparing \(\# m\)
for Fastind with that of \(\mathrm{Fastind}_{dim}\), we can observe the
effect of the use of GH polynomials on the template sizes.  Comparison
of \(t_{\mathit{sol}}\) for Fastind with that of
\(\mathrm{Fastind}_{dim}\) suggests the effect on the constraint
reduction phase; comparison of \(t_{\mathit{sol}}\) for Fastind with
\(t_{\mathit{inf}}+t_{\mathit{sol}}\) for \(\mathrm{Fastind}_{dim}\)
suggests the overhead incurred by g-degree inference.

\begin{wraptable}{l}{0.5\linewidth}
  \tiny
  \begin{center}
    \begin{tabular}{l|c|c|c|c|c|c|c}
      Name & & \multicolumn{2}{c|}{Fastind} & \multicolumn{4}{c}{Fastind\(_{dim}\)} \\
      \hline
      & deg & \(t_{\textit{sol}}\) & \(\# m\) & \(t_{\textit{inf}}\) & \(t_{\textit{sol}}\) &
      \(t_{\textit{inf}}+t_{\textit{sol}}\) & \(\# m\) \\
      \hline
      dijkstra & 2  & 9.29  & 21  & 0.456 & 8.83  & 9.29 & 21 \\
      divbin   & 2  & 0.674 & 21  & 0.388 & 0.362 & 0.750 & 8 \\
      freire1  & 2  & 0.267 & 10  & 0.252 & 0.258 & 0.510 & 10 \\
      freire2  & 3  & 2.51  & 35  & 0.463 & 2.60  & 3.06 & 35 \\
      cohencu  & 3  & 1.74  & 35  & 0.434 & 0.668 & 1.10 & 20 \\
      fermat   & 2  & 0.669 & 21  & 0.583 & 0.669 & 1.25 & 21 \\
      wensley  & 2  & 104   & 21  & 0.436 & 28.5  & 28.9 & 9 \\
      euclidex & 2  & 1.85  & 45  & 1.55  & 1.39  & 2.94 & 36 \\
      lcm      & 2  & 0.811 & 28  & 0.513 & 0.538 & 1.05 & 21 \\
      prod4    & 3  & 31.6  & 84  & 0.149 & 2.78  & 2.93 & 35 \\
      knuth    & 3  & 137   & 220 & 4.59  & 136   & 141 & 220 \\
      mannadiv & 2  & 0.749 & 21  & 0.515 & 0.700 & 1.22 & 18 \\
      petter1  & 2  & 0.132 & 6   & 0.200 & 0.132 & 0.332 & 6 \\
      petter2  & 3  & 0.520 & 20  & 0.226 & 0.278 & 0.504 & 6 \\
      petter3  & 4  & 1.56  & 35  & 0.226 & 0.279 & 0.505 & 7 \\
      petter4  & 5  & 7.15  & 56  & 0.240 & 0.441 & 0.681 & 8 \\
      petter5  & 6  & 17.2  & 84  & 0.228 & 0.326 & 0.554 & 9 \\
      petter10 & 11 & 485   & 364 & 0.225 & 0.354 & 0.579 & 14 \\
      sumpower1 & 3 & 2.20  & 35  & 0.489 & 2.31  & 2.80 & 35 \\
      sumpower5 & 7 & 670   & 330 & 0.469 & 89.1  & 89.6 & 140\\
      \hline
    \end{tabular}
  \end{center}
  \caption{Experimental result.}
  \label{table:result}
\end{wraptable}

\paragraph{Discussion}

The size of the templates, measured as the number of monomials (\(\#
m\)), was reduced in 13 out of 20 programs by using GH polynomials.  The
value of \(t_{\mathit{sol}}\) decreased for these 13 programs; it is
almost the same for the other programs.  \(\# m\) did not decrease for
the other seven programs because the extension of the type inference
procedure mentioned above introduced useless auxiliary variables.  We
expect that such variables can be eliminated by using a more elaborate
program analysis.

By comparing \(t_{\mathit{sol}}\) for Fastind and
\(t_{\mathit{inf}}+t_{\mathit{sol}}\) for
\(\mathrm{Fastind}_{\mathit{dim}}\), we can observe that the inference
of the g-degree assignment sometimes incurs an overhead for the entire
execution time if the template generated by Fastind is sufficiently
small; therefore, Fastind is already efficient.  However, this overhead
is compensated in the programs for which Fastind requires more
computation time.

To summarize, our current approach is especially effective for a program
for which (1) the existing invariant-synthesis algorithm is less
efficient owing to the large size of the template and (2) a nontrivial
g-degree assignment can be inferred.  We expect that our approach will
be effective for a wider range of programs if we find a more competent
g-degree inference algorithm.



\section{Related work}
\label{sec:relatedwork}




The template-based algebraic invariant synthesis proposed to
date~\cite{DBLP:conf/popl/SankaranarayananSM04,DBLP:journals/scp/CacheraJJK14}
has focused on reducing the problem to constraint solving and solving
the generated constraints efficiently; strategies for generating a
template have not been the main issue.  A popular strategy for template
synthesis is to iteratively increase the degree of a template.  This
strategy suffers from an increase in the size of a template in the
iterations when the degree is high.

Our claim is that prior analysis of a program effectively reduces the
size of a template; we used the dimension type system for this purpose
in this paper inspired by the principle of quantity dimensions in the
area of physics.  Of course, there is a tradeoff between the cost of the
analysis and its effect on the template-size reduction; our experiments
suggest that the cost of dimension type inference is reasonable.

\emph{Semialgebraic invariants} (i.e., invariants written using
\emph{inequalities} on polynomials) are often useful for program
verification.  The template-based approach is also popular in
semialgebraic invariant synthesis.  One popular strategy in
template-based semialgebraic invariant synthesis is to reduce this
problem to one of semidefinite programming, for which many efficient
solvers are widely available.


As of this writing, it is an open problem whether our idea regarding GH
polynomials also applies to semialgebraic invariant synthesis; for
physically meaningful programs, at least, we guess that it is reasonable
to use GH polynomials because of the success of the quantity dimension
principle in the area of physics.  A possible approach to this problem
would be to investigate the relationship between GH polynomials and
Stengle's Postivstellensatz~\cite{stengle1974nullstellensatz}, which is
the theoretical foundation of the semidefinite-programming approach
mentioned above.  There is a homogeneous version of the Stengle's
Positivstellensatz~\cite[Theorem~II.2]{SmoothPositivstellensatz};
because the notion of homogeneity considered there is equivalent to
generalized homogeneity introduced in this paper, we conjecture that
this theorem provides a theoretical foundation of an approach to
semialgebraic invariant synthesis using GH polynomials.

Although the application of the quantity dimension principle to
program verification is novel, this principle has been a handy tool
for discovering hidden knowledge about a physical system.  A
well-known example in the field of hydrodynamics is the motion of a
fluid in a pipe~\cite{barenblatt1996scaling}.  One fundamental result
in this regard is that of Buckingham~\cite{PhysRev.4.345}, who stated
that \emph{any physically meaningful relationship among \(n\)
  quantities can be rewritten as one among \(n-r\) independent
  dimensionless quantities, where \(r\) is the number of the
  quantities of the base dimension.}  Investigating the implications
of this theorem in the context of our work is an important direction
for future work.

The term ``generalized homogeneity'' appears in various areas; according
to Hankey et al.~\cite{hankey1972systematic}, a function
\(f(x_1,\dots,x_n)\) is said to be generalized homogeneous if there are
\(a_1,\dots,a_n\) and \(a_f\) such that, for any positive \(\lambda\),
\(f(\lambda^{a_1}x_1,\dots,\lambda^{a_n}x_n) = \lambda^{a_f}
f(x_1,\dots,x_n)\).  Barenblatt~\cite{barenblatt1996scaling} points out
that the essence of the quantity dimension principle is generalized
homogeneity.  Although we believe our GH polynomials are related to the
standard definition, we have not fully investigated the relationship at
the time of writing.

\iffull Our idea (and the quantity dimension principle) seems to be
related to \emph{invariant theory}~\cite{neusel00:_invar} in
mathematics.  Invariant theory studies various mathematical structures
using invariant polynomials.  A well-known fact is that a ring of
invariants is generated by homogeneous polynomials~\cite[Chapter
7]{Cox:2007:IVA:1204670}; GH polynomials can be seen as a generalization
of the notion of degree.

The structure of \(\POLY[x_1,\dots,x_n]\) resulting from the notion of
the generalized degrees is an instance of \emph{graded rings} from ring
theory.  Concretely, \(R\) is said to be \emph{graded} over an Abelian
group \(\GRP\) if \(R\) is decomposed into the direct sum of a family of
additive subgroups \(\set{R_g \mid g \in \GRP}\) and these subgroups
satisfy $R_g \cdot R_h \subseteq R_{gh}$ for all $g, h \in \GRP$.  Then,
an element \(x \in R\) is said to be \emph{homogeneous of degree \(g\)}
if \(x \in R_g\).  We leave an investigation of how our method can be
viewed in this abstract setting as future work.  \else \fi


\section{Conclusion}
\label{sec:conclusion}

We presented a technique to reduce the size of a template used in
template-based invariant-synthesis algorithms.  Our technique is based
on the finding that, if an algebraic invariant of a program \(c\)
exists, then there is a GH invariant of \(c\); hence, we can reduce the
size of a template by synthesizing only a GH polynomial.  We presented
the theoretical development as a modification of the framework proposed
by Cachera et al. and empirically confirmed the effect of our approach
using the benchmark used by Cachera et al.  Although we used the
framework of Cachera et al. as a baseline, we believe that we can apply
our idea to the other template-based
methods~\cite{DBLP:conf/cav/0001LMN14,DBLP:conf/fmcad/SomenziB11,DBLP:conf/popl/SankaranarayananSM04,DBLP:journals/jsc/Rodriguez-CarbonellK07,DBLP:journals/ipl/Muller-OlmS04,DBLP:journals/scp/CacheraJJK14,DBLP:conf/sas/AdjeGM15}.


Our motivation behind the current work is safety verification of hybrid
systems, in which the template method is a popular strategy.  For
example, Gulwani et al.~\cite{DBLP:conf/cav/GulwaniT08} proposed a
method of reducing the safety condition of a hybrid system to
constraints on the parameters of a template by using Lie derivatives.
We expect our idea to be useful for expediting these verification
procedures.

In this regard, Suenaga et
al.~\cite{DBLP:conf/popl/SuenagaSH13,DBLP:conf/cav/HasuoS12,DBLP:conf/icalp/SuenagaH11}
have recently proposed a framework called \emph{nonstandard static
analysis}, in which one models the continuous behavior of a system as an
imperative or a stream-processing program using an \emph{infinitesimal}
value.  An advantage of modeling in this framework is that we can apply
program verification tools without an extension for dealing with
continuous dynamics.  However, their approach requires highly nonlinear
invariants for verification.  This makes it difficult to apply existing
tools, which do not handle nonlinear expressions well.  We expect that
the current technique will address this difficulty with their framework.

We are also interested in applying our idea to decision procedures and
satisfiability modulo theories (SMT) solvers.  Support of nonlinear
predicates is an emerging trend in many SMT solvers (e.g.,
Z3~\cite{DBLP:conf/tacas/MouraB08}).  Dai et
al.~\cite{DBLP:conf/cav/DaiXZ13} proposed an algorithm for generating
a semialgebraic Craig interpolant using semidefinite
programming~\cite{DBLP:conf/cav/DaiXZ13}.  Application of our approach
to these method is an interesting direction for future work.


\section*{Acknowledgment}

We appreciate annonymous reviewers, Toshimitsu Ushio, Naoki Kobayashi
and Atsushi Igarashi for their comments.  This work is partially
supported by JST PRESTO, JST CREST, KAKENHI 70633692, and in
collaboration with the Toyota Motor Corporation.

\bibliographystyle{splncs03}
\bibliography{main}

\iffull
\appendix

\section{Proof of Theorem~\ref{thm:hconstr-constr-sound}}


To prove Theorem~\ref{thm:hconstr-constr-sound}, we define
\emph{renaming} of parameters and constraints.

\begin{definition}
  For an injection \(\iota : A \ra A'\), we write \(\iota : (A,G,C)
  \HLE (A',G',C')\) if \(G' = \iota^*(G)\) and \(C' = \iota^*(C)\)
  where \(\iota^*\) maps \(a' \in \iota(A)\) to \(\iota^{-1}(a')\) and
  \(a' \in \iota(A' \backslash \iota(A))\) to \(0\).
\end{definition}

The injection \(\iota\) gives a renaming of parameters.  The relation
\(\iota : (A,G,C) \HLE (A',G',C')\) reads \(G\) and \(C\) are obtained
from \(G'\) and \(C'\) by renaming the parameters in \(\iota(A)\)
using \(\iota\) and substituting \(0\) to those not in \(\iota(A)\).

\begin{lemma}
  If \(\iota : (A,G,C) \HLE (A',G',C')\), then there exists \(\kappa\)
  such that (1) \(\kappa : \WPCH{c}{\REMP,\Gamma}(A,G,C) \HLE
  \WPC{c}{\REMP}(A',G',C')\) and (2) \(\kappa\) is an extension of
  \(\iota\).
  
  
  \label{lem:constr-monotone}
\end{lemma}
\begin{proof}
  Induction on the structure of \(c\).
  \qed
\end{proof}


\begin{pfof}{Theorem~\ref{thm:hconstr-constr-sound}}
  Let $g \in T(A_0)$ be the most general template of generalized degree
  $\tau$ and degree $d$ and $g' \in T(A_0')$ be the most general
  template of degree $d$.  Without loss of generality, we assume \(A_0
  \subseteq A_0'\) and \(g' = g + g_1\) for some \(g_1 \in T(A_0'
  \backslash A_0)\).  Let $(A, G, C) = \WPCH{c}{\REMP,\Gamma}(A_0,
  \set{g}, \emptyset)$ and $(A', G', C') = \WPC{c}{\REMP}(A_0',
  \set{g'}, \emptyset)$.  Then, from Lemma~\ref{lem:constr-monotone},
  there exists \(\kappa\) such that \(\kappa : (A, G, C) \HLE (A', G',
  C')\) and \(\kappa\) is an extension of the inclusion mapping \(\iota
  : A_0 \ra A_0'\).  Suppose \(v(g)\) is a result of \(\INVINFH(c, d,
  \tau)\) where \(v\) is a solution to \(C \cup
  \set{\EQCONSTR{G}{\set{0}}}\).  
  Define a valuation \(v'\) on \(A'\) by
  \[
  v'(a') =
  \left\{
  \begin{array}{ll}
    v(a) & \mbox{\(a' = \kappa(a)\) for some \(a \in A\)}\\
    0 & \mbox{Otherwise.}
  \end{array}
  \right.
  \]
  Then, \(v'(g') = v'(g + g_1) = v'(g)\); the second equation holds
  because \(v'(a')\) is constantly \(0\) on any \(a' \in A' \backslash
  A\).  All the parameters in \(g\) are in \(A_0\) and \(\kappa\) is
  an identity on \(A_0\).  Therefore, \(v'(g) = v(g)\).
  It suffices to show that \(v' \models C' \cup
  \set{\EQCONSTR{G'}{\set{0}}}\), which indeed holds from the
  definition of \(v'\) since \(v \models C \cup
  \set{\EQCONSTR{G}{\set{0}}}\) and \(C\) and \(G\) are renaming of
  \(C'\) and \(G'\).
\end{pfof}

\else
\fi

\end{document}